\documentclass{article}

\usepackage{microtype}
\usepackage{graphicx}
\usepackage{booktabs} 
\usepackage{tabularx}
\usepackage{caption}
\usepackage{calc}
\usepackage{subcaption}
\usepackage{enumitem}
\usepackage{mathrsfs}
\usepackage{algorithmic}

\usepackage{hyperref}

\usepackage[accepted]{icml2024}

\usepackage{amsmath}
\usepackage{amssymb}
\usepackage{mathtools}
\usepackage{amsthm}

\usepackage[capitalize,noabbrev]{cleveref}

\theoremstyle{plain}
\newtheorem{theorem}{Theorem}[section]

\newtheorem{lemma}[theorem]{Lemma}

\theoremstyle{definition}

\theoremstyle{remark}

\usepackage[textsize=tiny]{todonotes}

\icmltitlerunning{Rethinking Specificity in SBDD}
\DeclareMathSizes{10}{10}{4}{3}
\begin{document}

\twocolumn[
\icmltitle{Rethinking Specificity in SBDD: \\Leveraging Delta Score and Energy-Guided Diffusion}



\icmlsetsymbol{equal}{*}

\begin{icmlauthorlist}
\icmlauthor{Bowen Gao}{equal,aaa}
\icmlauthor{Minsi Ren}{equal,bbb}
\icmlauthor{Yuyan Ni}{ccc}
\icmlauthor{Yanwen Huang}{ddd}
\icmlauthor{Bo Qiang}{ddd}
\icmlauthor{Zhi-Ming Ma}{ccc}
\icmlauthor{Wei-Ying Ma}{aaa}
\icmlauthor{Yanyan Lan}{aaa}

\end{icmlauthorlist}

\icmlaffiliation{aaa}{Institute for AI Industry Research (AIR), Tsinghua University}
\icmlaffiliation{bbb}{Institute of Automation, Chinese Academy of Sciences}
\icmlaffiliation{ccc}{Academy of Mathematics and Systems Science, Chinese Academy of Sciences}
\icmlaffiliation{ddd}{Department of Pharmaceutical Science, Peking University}

\icmlcorrespondingauthor{Yanyan Lan}{lanyanyan@air.tsinghua.edu.cn}

\icmlkeywords{Machine Learning, ICML}

\vskip 0.3in
]



\printAffiliationsAndNotice{\icmlEqualContribution} 

\begin{abstract}
In the field of Structure-based Drug Design (SBDD), deep learning-based generative models have achieved outstanding performance in terms of docking score. However, further study shows that the existing molecular generative methods and docking scores both have lacked consideration in terms of specificity, which means that generated molecules bind to almost every protein pocket with high affinity. To address this, we introduce the \emph{Delta Score}, a new metric for evaluating the specificity of molecular binding. To further incorporate this insight for generation, we develop an innovative energy-guided approach using contrastive learning, with active compounds as decoys, to direct generative models toward creating molecules with high specificity. Our empirical results show that this method not only enhances the delta score but also maintains or improves traditional docking scores, successfully bridging the gap between SBDD and real-world needs.
\end{abstract}

\section{Introduction}

Structure-based drug design (SBDD) has become a pivotal strategy in creating novel therapeutic agents. This approach leverages the three-dimensional structural information of target receptors to generate drug-like small molecules that can bind to these receptors effectively. Recent advancements, particularly those in deep learning-based generative models, have shown promising results. Studies such as \citep{NEURIPS2021_31445061, long2022DESERT, guan3d, guan2023decompdiff, peng2022pocket2mol} demonstrate the capability of these models to achieve docking scores that surpass those of reference ligands. 

However, a deeper analysis suggests that these accomplishments may not fully reflect practical efficacy. Further investigation of generated molecules with high docking scores reveals a critical issue: not only do they have a high docking score with the assigned pocket, but they also achieve high docking scores with other unintended protein pockets.
From a biological perspective, this observation indicates that current deep learning generative models generate molecules of low specificity.


Specificity is a crucial aspect of drug design. The efficacy of a drug is not solely dependent on its ability to bind to the intended disease-related target but also on its specificity for not binding to essential proteins. Promiscuous drugs, interacting with multiple biological targets, can potentially cause serious adverse effects\citep{harrison2016phase,lin2019off,wong2019estimation}. This lack of specificity contributes to the development of pan-assay interference compounds (PAINS), characterized by their undesirable broad biological activity and potential toxicity \citep{schneider2016novo}.


Unfortunately, specificity has been a largely overlooked factor in benchmarking SBDD previously. Currently, the most important metric for SBDD is based on docking scores. Traditional scoring functions merely evaluate the docking pose's spatial and electrostatic fit, failing to differentiate between highly specific molecules and promiscuous ones that bind to multiple targets \citep{zheng2022improving}. This limitation means modifications can artificially boost docking scores, misleadingly suggesting a compound's specificity and effectiveness. Thus, there's a pressing need for new metrics that accurately assess the specificity of molecule-pocket interactions, moving beyond the generalized interaction probability that current docking scores offer.

Besides benchmarking, they also fail to explicitly model the specific binding behavior for the generation. Generally, during training, these methods treat the pocket as context and attempt to reconstruct the ligand. 
Nevertheless, this method focuses on maximizing the joint probability of pocket-molecule pairs rather than the conditional probability of a molecule given a pocket, which more precisely reflects a drug's specificity and efficacy. We denote the conditional probability of a molecule given a pocket as the conditional binding probability.

To address these shortcomings, we have adopted a probabilistic framework and defined the conditional binding probability. From this perspective, we introduce a better evaluation metric named \emph{Delta Score} to measure the specific binding ability of current deep-learning-based SBDD models. The Delta Score has been designed and theoretically substantiated to reflect a molecule's selective binding behavior more accurately. It effectively bridges the theoretical predictions with the practical necessities of drug design.

Our analysis using the Delta Score has led to an essential finding: despite their prowess in achieving high docking scores, deep learning models underperform compared to reference ligands when evaluated using the Delta Score. This discrepancy highlights the need for improved deep-learning-based drug design methods.

In response to these challenges, we have developed an energy-guided approach\citep{dhariwal2021diffusion,zhao2022egsde,bao2022equivariant}. This approach is anchored in the use of contrastive learning and the strategic use of active compounds as decoys during training, inspired by \citep{Radford2021LearningTV,gao2023drugclip}. The training objective of our energy function aligns with the conditional binding probability formula, meaning that minimizing the training loss is equivalent to maximizing the conditional binding probability. This function guides the diffusion process in our generative model, directing it toward the conditional generation of molecular structures that exhibit high specific binding behavior.

Our experimental results are encouraging, demonstrating that our method not only improves the Delta Score but also maintains or enhances conventional docking scores. By implementing this novel approach, we aim to refine the process of AI-assisted drug design. Our goal is to ensure that the resulting molecules are not only effective in binding to their targets but also exhibit the selectivity required for safe and efficacious therapeutic use. 

Our contributions are summarized as follows:
\begin{itemize}
    \item The proposal of the Delta Score, an evaluation metric, measures the specific binding ability of generated molecules, demonstrating that advancements in previous methods predominantly enhance unconditional aspects rather than conditional binding.
    \item The introduction of a specific binding energy guidance approach, which explicitly models specific binding, directs the diffusion process toward producing molecules with a higher conditional binding probability to their designated targets.
    \item The achievement of promising empirical outcomes, which not only improve the Delta Score but also enhance the traditional docking score.
\end{itemize}

\section{Related Works}
With the emergence of geometric deep learning models \citep{satorras2021n,hoogeboom2022equivariant,geiger2022e3nn}, the field of Structure-Based Drug Design (SBDD) has shifted towards 3D neural networks for encoding protein structures and decoding 3D molecule conformations, representing real-world 3D interactions. Various methods have been proposed, including voxel-based methods \citep{masuda2020generating}, auto-regressive models \citep{NEURIPS2021_31445061,long2022DESERT}, and diffusion models \citep{guan3d,guan2023decompdiff}.

As for the evaluation metrics, most of the previous work used Vina score \citep{trott2010autodock} for binding affinity evaluation, which is a typical docking software to predict the interactions between a small molecule and a protein target. However, several studies have shown that other commercial docking software like Glide \citep{friesner2004glide} and Gold \citep{verdonk2003improved}, achieves a more accurate docking pose predicton and virtual screening ability compared to Vina. Recently, several studies have pointed out the issues of current benchmarking baselines, like PoseCheck\citep{harris2023benchmarking} and PoseBusters \citep{buttenschoen2024posebusters}. They mainly focus on the reliability of the conformations directly generated by the generative models. However, there is limited discussion on the reliability and value of the widely used docking score metric in the assessment of SBDD methods.

\begin{figure*}[ht]    
    \centering
    \begin{minipage}{0.4\textwidth}
        \centering
        \includegraphics[width=\linewidth]{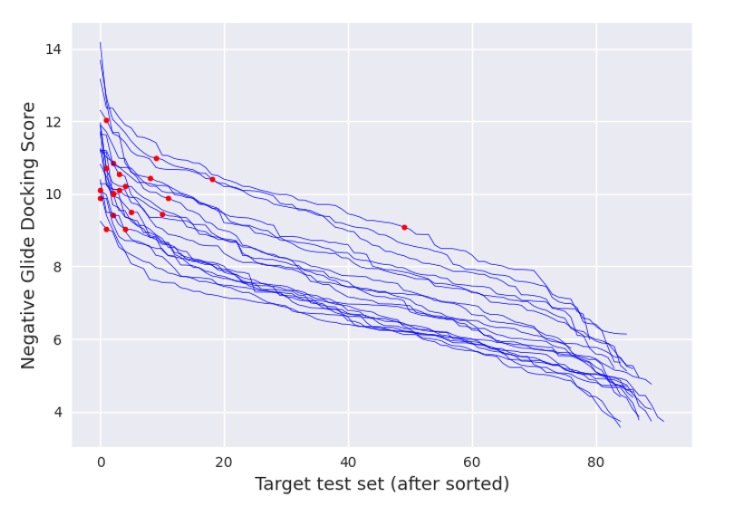}
        \captionsetup{singlelinecheck=false,justification=centering} 
        \subcaption{Docking Scores against targets in the test set after sorted (multiplied by -1)}    
        \label{fig:left}    
    \end{minipage}\hspace{0.01cm}    
    \begin{minipage}{0.4\textwidth}    
        \centering
        \includegraphics[width=\linewidth]{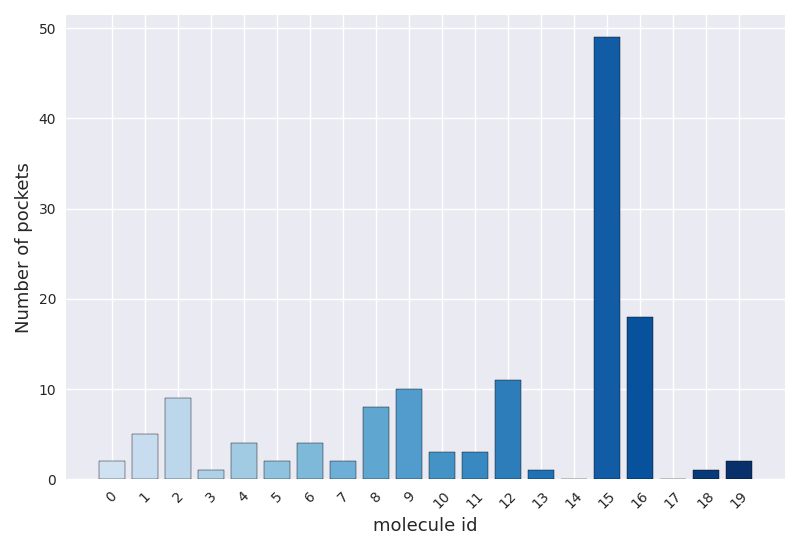}
        \captionsetup{justification=centering} 
        \subcaption{The number of pockets which surpass true target on docking score}    
        \label{fig:right}    
    \end{minipage}    
    \caption{In the left image, each line represents the sorted docking score of a small molecule against all pockets in the testset. The highlighted red represents its true target. The right image displays the number of test pockets for each small molecule in which the docking score is higher than their respective true target.}    
    \label{figure:superior scores}    
    \vspace{-0.5cm}
\end{figure*}
\vspace{0cm}
\section{Reassessing SBDD: Addressing Molecular Specificity}
We first conduct experiments on CrossDocked2020 dataset \citep{francoeur2020three} to compare different deep generative models, including auto-regressive model (denoted as AR) \citep{NEURIPS2021_31445061},  Pocket2mol \citep{peng2022pocket2mol}, DrugGPS \citep{zhang2023learning}, Targetdiff \citep{guan3d} and Decompdiff \citep{guan2023decompdiff}. The data preprocessing and splitting are all following \citet{guan3d} to ensure a fair comparison. For each method, we generate 20 molecules for each target pocket in the testset. Besides, we randomly sample 100 molecules for each target from the trainset, to form a baseline named `Random'. 
We adopt Glide which has demonstrated superior performance in predicting both accurate conformation and binding affinity to make our results more convincing \citep{zhang2023planet}.

\subsection{Notable Findings in Docking Scores}
The analysis of Table 1 reveals several key observations: \\
1. All methods can generate molecules with docking scores surpassing those of reference ligands across different targets. 
2. Diffusion based methods consistently exceed the reference ligands in mean docking scores. 
3. Surprisingly, \textbf{random sampling} from the training set can yield a mean Vina score close to that of reference ligands, with 30.83\% of molecules performing better.

Table \ref{table:docking score} shows that modern deep generative models are highly effective in producing molecules that outperform reference ligands, often with a success rate of 50\% or more. This suggests significant progress in Structure-Based Drug Design (SBDD). However, the reliability of docking scores as a sole quality measure is questionable. Research points to the susceptibility of docking software to manipulation, and our findings support this — about 30\% of randomly selected molecules from a molecular library can outperform reference ligands in docking scores, raising concerns about the current docking score metrics' validity.

\begin{table}[ht]
\vspace{-0.3cm}
\begin{center}
\caption{Docking Score Experimental Results. For a pocket with multiple generated molecules, we show results with both mean and min aggregation values.}
\label{table:docking score}


\resizebox{\linewidth}{!}{
\begin{tabular}{l|ccc}
\hline
& \multicolumn{2}{c}{Docking Scores($\downarrow$)} &  better than Ref ($\uparrow$) \\ \hline
Methods           & mean  & min & \\ \hline
Random          & -5.727   &-8.886    &30.83\%    \\
AR                &-5.738 & -7.529   &30.55\% \\
Pocket2Mol        & -5.945    &-7.734    &32.48\%  \\
DrugGPS    &-6.554 & -8.732  & 48.13\% \\
TargetDiff        & -6.664    &-8.501   &50.06\%  \\ 
DecompDiff & -7.102   &-8.970 & 56.04\% \\ \hline
Reference ligand  & -6.632     &-6.632   &-\\ \hline
\end{tabular}
}
\end{center}
\vspace{-0.5cm}
\end{table}

\subsection{Specificity Concerns in SBDD}

Our examination of molecules \textbf{with superior scores} from current models reveals a significant issue with specificity. We generated 100 molecules for each of 20 randomly selected pockets in the test set and chose the one with the best docking score for each pocket, resulting in 20 molecules. Despite their high glide docking scores (over -9) on their intended targets, cross-docking these molecules with all test set pockets (figure \ref{figure:superior scores}) showed that only 2 out of 20 had the highest docking score for their true targets. This indicates that these molecules generally bind well to multiple pockets, not just their specific target, highlighting a lack of specificity.

Specificity is crucial in drug design; non-specific molecules can bind to unintended targets, potentially causing adverse effects and diminishing efficacy. This issue is compounded by pseudo-active sites (PAINS), which can falsely enhance docking scores and mislead predictions of a molecule's true target. The prevalence of non-specific binding not only compromises drug effectiveness but also raises safety concerns due to potential side effects. High specificity in drug development is essential to ensure selective binding, maximizing therapeutic benefit and minimizing risks.

Both previous models and evaluation metrics have overlooked this critical aspect of specificity.

\subsubsection{Inadequacy of Docking Scores to Reflect Specificity}
Docking software, which is widely used in computational drug design, typically relies on force field models to calculate the interaction energy between molecules and receptors\citep{ferreira2015molecular,lyu2023modeling,mysinger2010rapid}. These models focus on the joint binding energy of pocket-ligand pairs, denoted as \( E(x, y) \). However, this approach does not fully address the specificity of binding interactions. 

A significant drawback of force field models lies in their empirical design, which relies on training with existing molecular structures and properties \citep{zheng2022improving, kitchen2004docking, quiroga2016vinardo}. This approach introduces biases, limiting the models' ability to accurately capture the full spectrum of molecular features and interactions. Consequently, when generative models produce molecules with structures that docking software erroneously favors, it leads to artificially inflated docking scores. This reveals a significant limitation of docking scores: they do not always accurately reflect the specific binding behavior of protein-molecule pairs. Instead, these scores can often reflect a bias towards certain molecular structures, misleadingly suggesting an unconditional binding ability \citep{eldridge1997empirical, wang2002further}.

\subsubsection{Limitations in Modeling Specificity in Previous SBDD Methods}

In the context of structure-based drug design (SBDD), the problem can be conceptualized as optimizing the generative model, denoted as \( \theta \), to maximize \( p_\theta(x|y) \). Here, \( x \) represents the molecule, while \( y \) denotes the protein pocket. Advanced generative models are adept at producing molecules with high docking scores. However, these molecules often demonstrate a lack of selectivity, exhibiting high docking scores across a wide range of pockets within the test set.

This phenomenon can be further elucidated through the Bayesian formula:
\begin{equation}
\label{bayesian}
    p_\theta(x,y) = p_\theta(y|x)p_\theta(x) \text{.}
\vspace{-0.2cm}
\end{equation}

We contend that the perceived progress made by contemporary advanced generative methods, evident in the enhanced \( p_\theta(x,y) \), is primarily attributable to improvements in the unconditional components, \( p_\theta(x) \). This development stands in stark contrast to the conditional component \( p_\theta(y|x) \), which is crucial for achieving specificity in molecular binding. A significant repercussion of this disparity is the tendency of generated molecules to attain uniformly high docking scores throughout the test set, irrespective of the distinctiveness of the pocket binding sites. The ensuing experimental results, to be detailed in the section \ref{section:result}, reinforce our hypothesis. 

\section{Method}

\subsection{From Unconditional to Conditional}

Our previous analysis underscores the need for a heightened focus on specific binding behaviors in SBDD to align with real-world requirements. In this section, we adopt a probabilistic framework to bridge the gap between the outputs of generative models and SBDD evaluation criteria. We use $x$ to denote a molecule and $y$ to denote a protein pocket from pre-defined finite pocket set $\mathbb{Y}$ and ligand set $\mathbb{X}$. Utilizing the Boltzmann distribution, the binding energy induces a joint probability describing the probability of complex formation:
\begin{equation}
\label{boltzman}
    p(x,y) = \frac{1}{Z}e^{-E(x,y)},
\end{equation}
where $Z$ is the normalization coefficient.

%

Crucially, we focus on the conditional probability \( p(y|x) \), the likelihood that molecule \( x \) binds specifically to pocket \( y \) among all potential targets. We propose to formulate a novel specificity metric by using the logarithm of the conditional probability, the larger the better:
\begin{align}\label{eq:conditional measure main}
    &\log p(y|x) = \log\frac{p(x,y)}{p(x)} = \log\frac{\frac{1}{Z}e^{-E(x,y)}}{ \sum_{y\in \mathbb{Y}}\frac{1}{Z}e^{-E(x,y)}} \notag\\
    &=  -E(x,y)-\log\sum_{y\in \mathbb{Y}} e^{-E(x,y)}.     
\end{align}

Equation \eqref{eq:conditional measure main} indicates that to assess the specific binding ability, it is necessary to consider the docking scores of both positive pairs and negative pairs rather than solely focusing on the score of one pair of molecule and receptor. 
This approach allows us to quantitatively assess the specificity of molecule-pocket interactions from the pairwise interaction energy, highlighting a crucial factor for the efficacy of SBDD strategies.

\subsection{Measure the specific binding}
Inspired by the general specificity metric defined in the previous section, we propose the \textbf{Delta Score} to measure the generated model and reflect the specificity and efficacy of the generated molecules. 

For a dataset with \( n \) pockets \( y_1, y_2, \cdots, y_n \),
the model to be evaluated generates \( m_i \) molecules \( x_{i1}, x_{i2}, \cdots, x_{im_i} \) for each pocket \( y_i \). Then for any $i\in\{1,\cdots, n\}$, \(j\in\{1,\cdots, m_i\}\), the specificity metric can be calculated by
\vspace{-0.15cm}
\begin{align}
&S(x_{ij},y_i)+\log\sum_{k=1}^n e^{-S(x_{ij},y_k)} \label{eq:ds11}\\
  & \approx S(x_{ij},y_i)-\min_{y_k} S(x_{ij},y_k) ,\label{eq:ds22}
\end{align}
where we use a docking score, denoted as $S(x,y)$, to estimate the interaction energy $E(x,y)$. \eqref{eq:ds11} is the negative of the metric defined in \eqref{eq:conditional measure main}. \eqref{eq:ds22} is a tight approximation of \eqref{eq:ds11}, with the log-sum-exp approximated by the max function. As a result, \eqref{eq:ds22} has the same unit as the utilized docking score $S(x, y)$ and is also the smaller the better, analogous to existing docking scores.

 
 Notice that in order to calculate the minimization term in \eqref{eq:ds22} for all $i=1\cdots n$, $j=1,\cdots m_i$, we need to perform docking for molecules with all possible pockets which is of quadratic complexity $O(n^2)$. To address this challenge, we provide a computational efficiency estimation of \eqref{eq:ds22} by random sampling. Finally, we derive the Delta Score as the specificity metric of the generative model for a given pocket $y_i$:
\begin{equation}
\label{formula:delta score main}
\begin{split}
\textbf{Delta Score}(y_i) = \frac{1}{m_i} \sum_{j=1}^{m_i} (-S(x_{ij},y_i) + S(x_{ij},y_k)), \\
 \text{for each } i, \text{we sample } k \in \{1, 2, \ldots, n\} \text{ with } k \neq i.
\end{split}
\end{equation}
Intuitively, a smaller value of Delta Score indicates the molecules generated by the model specifically for target $y_i$, have a high affinity towards target $y_i$ itself relative to the average affinity across all the other targets $y_k$.
We provide more discussion of the derivation of Delta Score in Appendix \ref{sec:app DS}.

\subsection{Generation with Specific binding}
\subsubsection{Overview of our method}

\begin{figure*}[ht]
\begin{center}
\centerline{\includegraphics[width=1\linewidth]{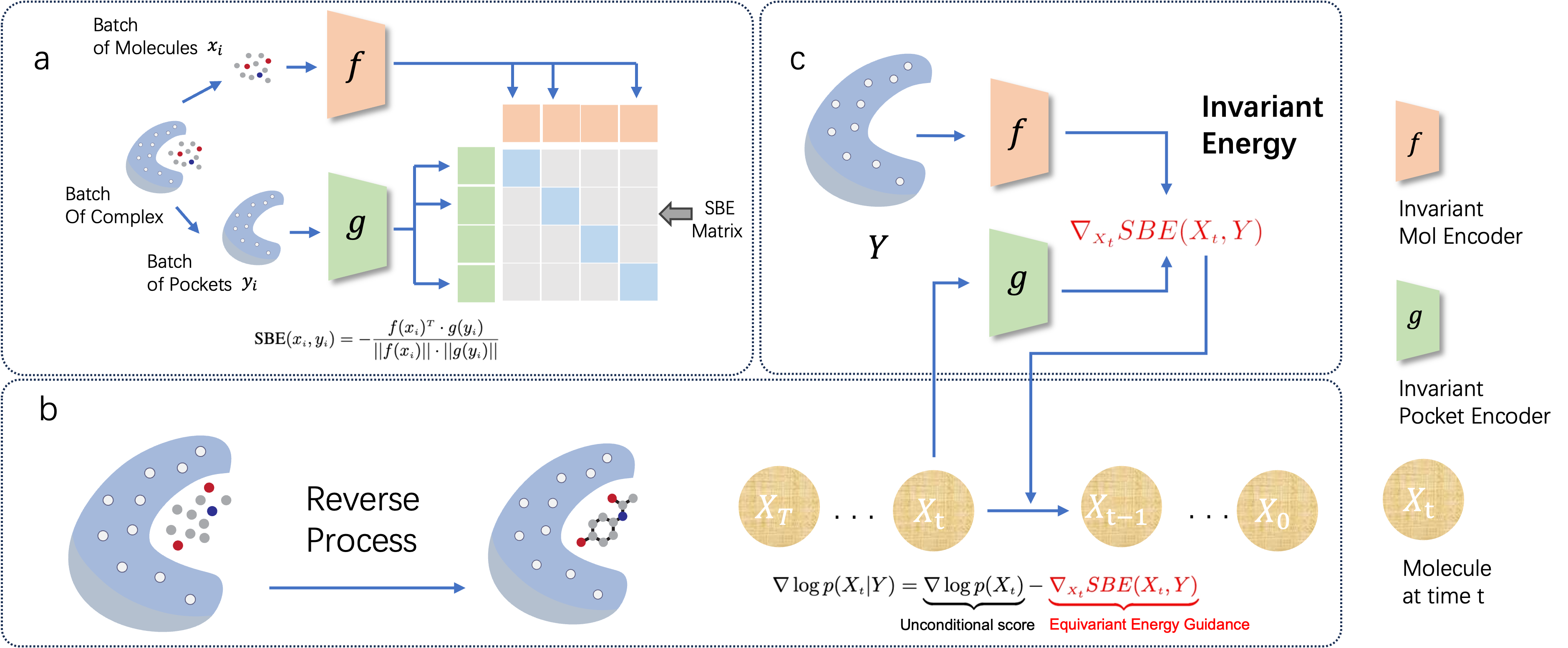}}
\caption{The main framework of our SBE-Diff model. Part a) depicts the training phase of the Specific Binding Energy (SBE) model, where in-batch softmax loss is employed. We use ligands designated for different pockets as negative instances for the current pocket, aiming to address the specific binding. Part b) describes the model's reverse diffusion process, wherein the transition from \(X_{t+1}\) to \(X_t\) is guided by minimizing the specific binding energy. Part c) demonstrates the calculation of the SBE. This is achieved by utilizing pocket and molecule embeddings, produced by their respective specially trained encoders.}
\label{figure: main framework}
\end{center}
\vspace{-0.7cm}
\end{figure*}
We denote the molecule as $X = [c^X, v^X]$, where [·, ·] is the concatenation operator and $c \in R^{N_X\times3}$ and $ v \in R^{N_X\times K}$ denote atom Cartesian coordinates and one-hot atom types respectively, $N_X$ is the number of atoms, $K$ is the number of type of an atom respectively. Pockets are denoted as $Y = [c^Y, v^Y]$  similarly. 


The overview of our method is shown in figure \ref{figure: main framework}. Our primary objective is to enhance the specificity of the interaction between synthesized small molecules and their intended targets. At the core of our methodology, we develop a robust model that is adept at assessing this particular attribute. The trained model itself functions as an energy function, which plays a crucial role in guiding the diffusion-based synthesis of these molecules. By effectively utilizing this function, we can direct the generation of molecules towards those that demonstrate a heightened ability to bind specifically to their targets. A key requirement for the molecular generative process is that the probability $p_\theta(x_0 | y)$ should be invariant to translation and rotation of the protein-ligand complex, which will be satisfied if the following properties hold \citep{guan3d,bao2022equivariant}:
\begin{enumerate}
  \item The initial distribution of diffusion process $p_\theta(X_T|Y) $ is an SE(3)-invariant distribution.
  \item The transformation at each time step of the diffusion process is SE(3)-equivariant.
  \item The energy function used to guide the diffusion process is invariant under SE(3) transformations.
\end{enumerate}
To meet the above requirements, we move the center of mass (CoM) of the protein atoms to zero during initialization, and apply isotropic Gaussian noise as the coordinates of atoms at time step T, ensuring that the initial distribution is SE(3)-invariant. Furthermore, we adopt 3D-Equivariant graph neural network to parameterize the denoiser for molecular diffusion process as well as the encoder for the energy function. We provide detailed information on the feature updates of the SE(3)-Equivariant GNN network in Appendix B.
\subsubsection{Molecular Diffusion Model}
Following \citep{guan3d}, in the forward diffusion process, we define the transition kernels of the diffusion model for atomic coordinates and types using Gaussian noise and categorical noise, respectively. $\beta_t$ is defined by fixed variance schedules, we denote $\alpha_t = 1 - \beta_t$, $\overline{\alpha}_t=\prod_{s=1}^t \alpha_s$ and $\Tilde{\beta_t}=\frac{1-\overline{\alpha}_{t-1}}{1-\overline{\alpha}_t}\beta_t$, we have:

\vspace{-0.0cm}
\begin{align}
q(X_t|X_{t-1}) &= \prod_{i=1}^{N_X}\mathcal{N}(c_{t,i}^X;\sqrt{\alpha_t}c_{t-1,i}^X,\beta_t\mathcal{I}) \notag \\
&\qquad \cdot \mathcal{C}(v_{t,i}^{X}|\alpha_t v_{t-1,i}^X + \beta_t/K)
\\
q(X_t|X_{0}) &= \prod_{i=1}^{N_X}\mathcal{N}(c_{t,i}^X;\sqrt{\overline{\alpha}_t}c_{0,i}^X,(1-\overline{\alpha}_t)\mathcal{I}) \notag \\
&\qquad \cdot
\mathcal{C}(v_{t,i}^{X}|\overline{\alpha}_t v_{0,i}^X + (1-\overline{\alpha}_t)/K)
\label{formula:forwar}
\end{align}


The corresponding normal posterior if atom coordinates and categorical posterior of atom types can be computed as:
\begin{align}
q(X_{t-1}|X_t, X_0) &= \prod_{i=1}^{N_X}\mathcal{N}(c_{t-1,i}^X;\Tilde{\mu}(c_{t,i}^X, c_{0,i}^X),\Tilde{\beta}_t\mathcal{I}) \notag \\
&\qquad \cdot \mathcal{C}(v_{t-1,i}^{X}|\Tilde{c}(v_{t,i}^X, v_{0,i}^X))\text{.}
\end{align}
\begin{align}
\Tilde{\mu}(c_{t,i}^X, c_{0,i}^X)=\frac{\sqrt{\overline{\alpha}_{t-1}}\beta_t}{1-\overline{\alpha}_t}&c_{0,i}^X + \frac{\sqrt{\alpha_t}(1-\overline{\alpha}_{t-1})}{1-\overline{\alpha}_t} c_{t,i}^X, \notag \\
c^*(v_{t,i}^X,v_{0,i}^X) &=  [\overline{\alpha}_{t-1}v_{0,i}^X + \frac{(1-\overline{\alpha}_{t-1})}{K}] \notag \\
&\odot[\alpha_t v_{t,i}^X + \frac{(1-\alpha_t)}{K}],\notag \\
\Tilde{c}(v_{t,i}^X,v_{0,i}^X) &= \frac{c^*}{\sum_{k=1}^Kc_k^*}\text{.}
\end{align}
During the generative process, we will recover molecule $X_0$ from the initial noise $X_T$. We predict the atom coordinates $\hat{c}_{0|t,i}^X$ and atom types $\hat{v}_{0|t,i}^X$ using denoising SE(3)-Equivariant network $\theta$ to approximate the reverse distribution at every time step $t$ as follows:
\begin{align}
p_\theta(X_{t-1}|X_t) &= \prod_{i=1}^{N_X}\mathcal{N}(c_{t-1,i}^X;\Tilde{\mu}(c_{t,i}^X, \hat{c}_{0|t,i}^X),\Tilde{\beta}_t\mathcal{I}) \notag \\
&\qquad \cdot \mathcal{C}(v_{t-1,i}^{X}|\Tilde{c}(v_{t,i}^X,\hat{v}_{0|t,i}^X))\text{.}
\end{align}
We also provide detailed information on the parameterization of the diffusion process in Appendix B.

\subsubsection{Specific Binding Energy Function}
Then we introduce the SBE-model we use to serve as the guided energy function. The SBE-model consists of a molecule encoder $\text{SBE-Enc}_X$ as well as a pocket encoder $\text{SBE-Enc}_Y$. For a protein pocket $y_i$ and a molecule $x_i$, we first define the specific binding energy (SBE) as:
\begin{equation}
\begin{aligned}
 &\text{SBE}(x_{i}, y_i) = -\frac{f(x_{i})^{T} \cdot g(y_i)}{||f(x_{i})|| \cdot ||g(y_i)||}\text{ , where}\\
 &f(x_{i}) = \text{SBE-Enc}_X(X_i)\text{,  }g(y_i) =  \text{SBE-Enc}_Y(Y_i)\text{.}
\end{aligned}
\end{equation}

The smaller the SBE is, the stronger the binding specificity. 
In order to provide guidance for the generative process throughout all diffusion steps, we train this energy function to be time-dependent, which receives the noisy molecule data $x_{t,i}$ as inputs. During the training phase, we randomly sample time step $t$ and apply formula \ref{formula:forwar} to sample the noisy molecule data $x_{t,i}$ from ground truth $x_i$.
For a batch of size $N$, we have pairwise data $\{(x_i,y_i)\}^{N}_{i=1}$. we extract a list of protein pockets $\{y_i\}^{N}_{i=1}$ and a list of corresponding molecules  $\{x_i\}^{N}_{i=1}$. Combining them together results in $N^2$ pairs $(x_{i},y_j)$ where $i,j \in[1,N]$. When $i = j$ it is a positive pair, and when $i \neq j$ it is a negative pair.

Our first training goal is to identify the true binding molecules before other molecules for a given protein pocket:
\begin{equation}
   \mathcal{L}_{p}(y_i,\{x_{t,j}\}^{N}_{j=1}) = -\frac{1}{N}\log \frac{\exp(-\text{SBE}(x_{t,i},y_i)/\tau)}{\sum_{j} \exp(-\text{SBE}(x_{t,j},y_i)/\tau)}\text{,}
\end{equation}

Symmetrically, our second training goal is to identify the true binders from a batch of protein pockets.
\vspace{-0.2cm}
\begin{equation}
   \mathcal{L}_{m}(x_{i,t},\{y_j\}^{N}_{j=1}) = -\frac{1}{N}\log \frac{\exp(-\text{SBE}(x_{t,i},y_i)/\tau)}{\sum_{j} \exp(-\text{SBE}(x_{t,i},y_j)/\tau)}\text{,}
\end{equation}
From the theoretical side, our training loss is exactly the same format as the log conditional probability in Section 4.1. Minimizing the Pocket to Mol loss $\mathcal{L}_{p}$ is equivalent to maximizing $p(y|x)$, the probability that a molecule binds to its corresponding protein pocket specifically. Similarly, minimizing the Mol to Pocket loss $\mathcal{L}_{m}$ is equivalent to maximizing $p(x|y)$, the probability that a pocket binds to its corresponding ligand specifically.

\subsubsection{SB-Energy Minimizer Diffusion Model}
We propose an energy-guided diffusion model that dynamically adjusts the types and positions of atoms in the generative process, leveraging the SBE-model. We minimize this SB-energy using the techniques of the Energy-Guided Diffusion Model \citep{bao2022equivariant} to facilitate the ability of the generated molecule to specifically bind to its target:
\vspace{-0.0cm}
\begin{equation}
\begin{aligned}
 p(X|Y) \propto p(X)\cdot\exp{(-\mathcal{E}(X,Y))} \text{, }
\end{aligned}
\end{equation}
where we adopt SBE as the guided energy function $\mathcal{E}(X,Y)$. More concretely, at each time step of the reverse diffusion process, we feed the noisy molecule $X_t$ and its corresponding pocket $Y$ into the SBE-model to calculate the SB-Energy and update the denoising process using its gradient with respect to the molecule $\nabla_{X_t}\text{SBE}(X_t,Y)$:
\begin{align}
\hat{p}_\theta(X_{t-1}|X_t, Y) &= \notag \\
\prod_{i=1}^{N_X}\mathcal{N}(c_{t-1,i}^X;\Tilde{\mu}(c_{t,i}^X , \hat{c}_{0|t,i}^X)&\textcolor{blue}{-w_1\nabla_{c^X_t}\text{SBE}(X_t,Y)},\Tilde{\beta}_t\mathcal{I}) \notag \\
\qquad \cdot \mathcal{C}(v_{t-1,i}^{X}|\Tilde{c}(v_{t,i}^X,\hat{v}_{0|t,i}^X)&\textcolor{blue}{-w_2\nabla_{v^X_t}\text{SBE}(X_t,Y)}) \text{.}
\label{formula:sbe-guided}
\end{align}
where $w_1$ and $w_2$ are empirical coefficients used to adjust the intensity of the energy condition, we set $w_1$ to 0.1 and $w_2$ to 1e-4 in practical. 
The complete sampling procedure is shown in Algorithm 1.
\begin{algorithm}
  \caption{Sampling Procedure of SBE-Diff}
  \textbf{Input:} The pocket Y, the model $\phi_{\theta}$, coefficients $w_1$ $w_2$. \\
  \textbf{Output:} Generated ligand molecule X that binds to the protein pocket specifically.
\begin{algorithmic}[1]
    \STATE Sample the number of atoms $N_X$ based on a prior distribution conditioned on the pocket size.
    \STATE Model the whole complex as a graph $[X, Y]$.
    \STATE Move center of mass of protein atoms to zero.
    \STATE Sample initial molecular atom coordinates $c_T$ and atom types $v_T$: $c_T \in \mathcal{N}(0, I)$, $v_T = \text{one hot}(\arg \max_i g_i)$, where $g \sim \text{Gumbel}(0, 1)$.
    \FOR{$t$ in $T, T-1, \ldots, 1$}
        \STATE Predict $\hat{X}_0 = [\hat{c}_0^X, \hat{v}_0^X]$ from $c_t^X, v_t^X$ with $\phi_{\theta}$: $\hat{c}_0^X, \hat{v}_0^X = \phi_{\theta}(c_t^X, v_t^X, t, Y)$.
        \STATE Calculate the posterior $\hat{p}_{\theta}(X_{t-1}|X_t, Y)$ according to equation \ref{formula:sbe-guided}.
        \STATE Sample $c_{t-1}^X$ and $v_{t-1}^X$ from the posterior $\hat{p}_{\theta}$.
    \ENDFOR
\end{algorithmic}
\end{algorithm}
\vspace{-0.5cm}
\section{ Experiments}
\subsection{Main Results}
\label{section:result}

In our study, we conducted experiments on the CrossDocked2020 dataset, following the data split outlined by \citet{guan3d}. We use the CrossDocked2020 dataset to train both the generative model and the SBE model.

Given the necessity of performing docking after shuffling binding sites, we faced a unique challenge: using the original generated 3D conformations of molecules would be unfair and unreasonable, as a molecule created for one pocket should be docked to a different pocket for this experiment. To address this, we converted all generated molecules into 1d/2d representations(Rdkit Mols) and let Glide to generate initial conformations. These conformations were then docked, and the one with the highest docking score was selected.  

Beyond the conventional docking score, we introduce the delta score, as detailed in formula \ref{formula:delta score main}. Additionally, we present the ratio indicating instances where the generated molecule outperforms the reference ligand for a specific pocket. In this context, `better performance' is defined as the scenario where the generated molecule achieves both a higher absolute docking score and a better delta score compared to the corresponding reference ligand. For the aggregation of different binding sites, we show the results in mean and median values.


\begin{table}[ht]
\vspace{-0.3cm}
\begin{center}
\caption{Experimental Results for different methods. We show the absolute docking score, delta score, and ratio of the generated molecule better than the reference ligand. For each metric, the best result is \textbf{bold} and the second best result is \underline{underlined}.}
\label{table:main_results}
\resizebox{\linewidth}{!}{
\begin{tabular}{l|ccccc}
\hline
& \multicolumn{1}{c}{Absolute $\uparrow$} & \multicolumn{2}{c}{Delta $\uparrow$} & \multicolumn{2}{c}{Ratio $\uparrow$}\\ \hline
Methods           & mean   & mean   & median   & mean  &median   \\ \hline
trainset          &5.727   & 0.044   & -0.073     &0.159 &0.053          \\
AR                &  5.737  &0.393   &0.200 &0.150 &0.039    \\
Pocket2Mol        & 5.946    &0.437    &0.106 &0.170  &0.050         \\ 
DrugGPS           &6.554        &\underline{0.490}    &\textbf{0.387} &0.235  & 0.134  \\ \hline
TargetDiff        & 6.665    &0.335    &0.102 & 0.259  &0.129  \\ 
DecompDiff  & \textbf{7.102}           &0.354    &0.220 & \underline{0.274} & \underline{0.133}  \\ 
SBE-Diff &\underline{6.815}   &\textbf{0.552}   &\underline{0.250} &\textbf{0.300}  &\textbf{0.216} \\ \hline
Reference  & 6.632          &1.158    &1.029  & - &-   \\ \hline
\end{tabular}
}
\end{center}
\end{table}

\begin{figure}[ht]
\begin{center}
\centerline{\includegraphics[width=1.0\columnwidth]{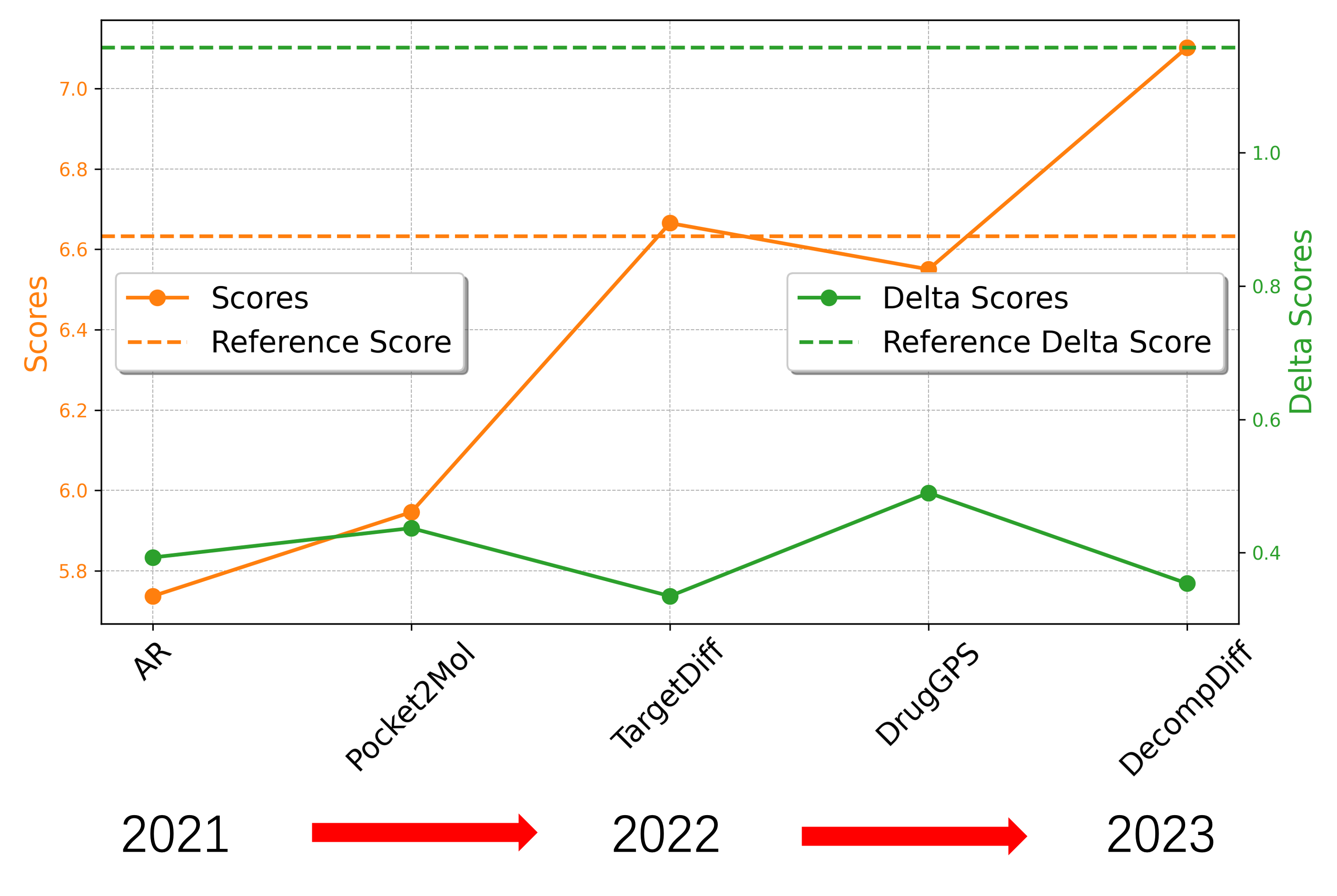}}
\caption{The evolution of absolute docking scores and delta scores obtained by various methods, organized chronologically.}
\label{figure:trend}
\end{center}
\vspace{-0.5cm}
\end{figure}

\begin{table*}[h]
\begin{center}
\caption{Comparison of SBE-Diff with TargetDiff in terms of different metrics commonly used in SBDD evaluation. Difference is shown in percentage.}
\label{table:ablation}
\resizebox{0.8\linewidth}{!}{
\begin{tabular}{l|ccccccc}
\hline
& ori score $\uparrow$ &shuffle score $\downarrow$  & delta score $\uparrow$ &Ratio $\uparrow$ & QED $\uparrow$ & SA $\uparrow$ & Diversity $\uparrow$\\ \hline
TargetDiff          &6.665   & 6.320   & 0.335   &0.259  &0.503 &0.587  &0.880        \\
SBE-Diff            &  6.815  &6.261   &0.552 &0.300 &0.497 &0.586 &0.879   \\ \hline
Difference          & 2.2\%   &0.9\% &64.8\% &15.8\% &-1.2\% &-0.0\%  & -0.0\% \\ \hline
\end{tabular}
}
\vspace{-0.5cm}
\end{center}
\end{table*}

Table \ref{table:main_results} and Figure \ref{figure:trend} collectively highlight a critical trend in structure-based drug design (SBDD). While several methods, including recent ones like Targetdiff and DecompDiff, achieve absolute docking scores surpassing those of reference ligands, they fall short in delta scores, indicating a deficiency in specific binding ability. This discrepancy is particularly pronounced when compared with older methods like AR and Pocket2Mol. The chronological plot in Figure \ref{figure:trend} further illustrates this point; despite a clear upward trajectory in docking scores over time, delta scores, represented by the green line, fluctuate without significant improvement and remain substantially lower than those of reference ligands. This observation underscores our argument that \textbf{current advancements in SBDD are primarily focused on improving unconditional part ($p_{\theta}(x)$), rather than the more crucial conditional aspects of molecular binding ($p_{\theta}(y|x)$).}

As shown in table \ref{table:main_results}, our SBE-Diff approach stands out remarkably in the comparative analysis, consistently ranking either first or second across all evaluated metrics, as detailed in our results. This performance is particularly noteworthy in terms of Delta Scores and the ratio of surpassing reference ligands, where our method demonstrates exceptional proficiency. Such results unequivocally indicate that the SBE-Diff approach excels in directing the generative process towards enhanced specificity in molecular binding. 

A notable observation from our study is the better performance of DecompDiff over TargetDiff in terms of conditional binding effectiveness. This outcome is likely attributed to DecompDiff's innovative approach of utilizing pocket information to establish priors for atom clustering, thereby explicitly and effectively modeling conditional binding. Similarly, DrugGPS demonstrates remarkable capabilities in specific binding scenarios. This is largely due to its unique strategy of constructing a subpocket prototype-molecular motif interaction graph, which directly and proficiently models the conditional binding aspect. These findings underscore the importance of tailored approaches in enhancing conditional binding effectiveness in drug design.

\subsection{Ablation and Analysis}
\begin{figure}[ht]
\begin{center}
\centerline{\includegraphics[width=1\columnwidth]{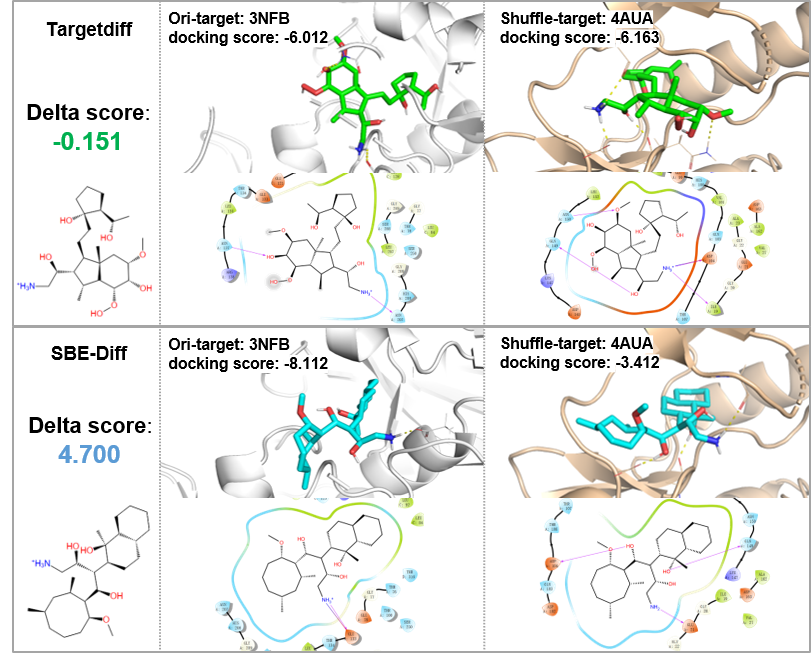}}
\caption{Case study of generated molecules docked with its original target and shuffled target.}
\label{figure:case}
\end{center}
\vspace{-0.5cm}
\end{figure}
\subsubsection{Effectiveness of the SBE guidance}

Our diffusion process and model bear similarities to Targetdiff, with a notable distinction being our Specific Binding Energy (SBE) guidance component. To demonstrate the effectiveness of our SBE guidance framework, we conducted a comparative analysis with Targetdiff. Alongside docking scores, we included other common SBDD benchmarking metrics in our evaluation.

The results, detailed in Table \ref{table:ablation}, show that our SBE-Diff approach positively impacts specificity without compromising other key metrics like QED, Synthesizability(SA), and diversity. Notably, our method enhances the absolute docking score against the specific target and reduces the score post-shuffling. This leads to a significant increase in the delta score, effectively capturing the improvement in specificity. 

\subsubsection{Visualization}

A case of the molecule generated by SBE-Diff versus the molecule generated by TargetDiff is shown in figure \ref{figure:case}. The SBE-Diff molecule achieves a significantly higher Delta score of 4.700, in contrast to the Targetdiff molecule with a modest Delta score of -0.151. The Targetdiff molecule exhibits similar docking scores on both ori-target and shuffle-target. This observation aligns with our understanding that molecules rich in hydroxyl groups, e.g. sugars and polyols, can bind to various protein pockets through numerous hydrogen bonds and other interaction mechanisms. SBE-Diff performs well, achieving a docking score of -8.112 on ori-target. SBE-Diff does this by reducing these unnecessary hydroxyl groups while adding flexible hydrophobic moiety to the small molecule, increasing the contribution of the entropy effect of the small molecule binding process. Interestingly, SBE-Diff also removes the problematic peroxide group present in the original molecule. SBE-Diff helps increase molecule specificity by reducing potential affinity for the shuffle-target, while also improving the rationality of binding pose to the intended target. 
We show more cases in appendix \ref{appendix: visualization}.

\section{Conclusion}

In our study, we analyze and rethink the current state of structure-based drug design, identifying a significant focus on generating molecules with high docking scores that primarily address unconditional aspects. We note a considerable lack of emphasis on specificity, which is crucial in effective drug design. To tackle this issue, we introduce a new evaluation metric and a guidance framework centered on conditional binding to enhance specificity. Our experimental results validate our observation about the field's inclination towards unconditional aspects and demonstrate the effectiveness of our proposed methods in improving specificity. Our work not only advocates for a paradigm shift in the approach to drug design but also provides a concrete and effective direction for achieving this goal.

\newpage

\paragraph{Impact Statement}
This paper presents work whose goal is to advance the field of Machine Learning. There are many potential societal consequences of our work, none of which we feel must be specifically highlighted here.
\newpage

\bibliography{paper}
\bibliographystyle{icml2024}

\newpage
\appendix
\onecolumn
\section{Delta Score}\label{sec:app DS}
\subsection{General Specificity Metric}
Existing docking scores typically calculate the interaction energy between the molecule and the protein pocket to reflect the binding strength~\cite{friesner2004glide}. We denote the interaction energy as $E(x,y)$, with the molecule and the pocket denoted as \( x \) and \( y \), respectively. Utilizing the Boltzmann distribution, the energy naturally defines a joint probability describing the probability that \( x \) and \( y \) can bind with each other.
\begin{equation}\label{eq:boltzman}
    p(x,y) = \frac{1}{Z}e^{-E(x,y)}, 
\end{equation}
where $Z=\sum_{x\in \mathbb{X}} \sum_{y\in \mathbb{Y}} e^{-E(x,y)}$ is the normalization coefficient and the sum is over the pre-defined finite pocket set $\mathbb{Y}$ and ligand set $\mathbb{X}$.
As a result, existing docking scores solely describe the joint probability of pocket-ligand pairs and are unable to reflect the specificity of binding interactions \( p(y | x) \), which describes the probability that molecule $x$ binds specifically to pocket $y$ among all potential targets.

To address the challenge that existing docking scores overlook binding specificity, we propose to focus on the conditional probability \( p(y|x) \) and formulate a novel specificity metric. Consider the logarithm of the specific binding probability for a molecule $x$ and a pocket $y$:
\begin{align}\label{eq:conditional measure}
    &\log p(y|x) = \log\frac{p(x,y)}{p(x)} = \log\frac{\frac{1}{Z}e^{-E(x,y)}}{ \sum_{y\in \mathbb{Y}}\frac{1}{Z}e^{-E(x,y)}} =  -E(x,y)-\log\sum_{y\in \mathbb{Y}} e^{-E(x,y)}.     
\end{align}
Therefore, \eqref{eq:conditional measure} functions as a specificity metric of binding and can be calculated from the pairwise interaction energy.

\subsection{Delta Score: A Specificity Metric for Generated Models}\label{sec:app DS proof}
In SBDD benchmarking, conventional methods use docking scores to measure the quality of the generated models and which fails to explicitly embody the specific binding behavior. Inspired by the general specificity metric defined in the previous section, we propose the Delta Score to measure the generated model and reflect the specificity and efficacy of the generated molecules. 

For a dataset with \( n \) pockets \( y_1, y_2, \cdots, y_n \),
the model generates \( m_i \) molecules \( x_{i1}, x_{i2}, \cdots, x_{im_i} \) for each pocket \( y_i \). Then for any $i\in\{1,\cdots, n\}$, \(j\in\{1,\cdots, m_i\}\), specificity metric can be calculated by
\begin{align}\label{eq:ds1}
    -S(x_{ij},y_i)-\log\sum_{k=1}^n e^{-S(x_{ij},y_k)} ,
\end{align}
where we use a docking score, denoted as $S(x,y)$, to estimate the interaction energy $E(x,y)$.

Next, we provide a simpler form of the specificity metric by using the following lemma. 
\begin{lemma}[Log-sum-exp approximation]
 For a finite set of real numbers \( z_1, z_2, \cdots, z_N \), the log-sum-exp function can be upper and lower bounded by the maximum function.
    \begin{equation}
        \max\{z_1,\cdots, z_N\}\leq \log(\sum_{k=1}^N\exp(z_k))\leq \max\{z_1,\cdots, z_N\}+\log(N) \text{.}
    \end{equation}
\end{lemma}
\begin{proof}
\begin{equation}
    \begin{aligned}
    \max\{z_1,\cdots, z_N\}&=\log(\exp(\max\{z_1,\cdots, z_N\}))\\
    &\leq \log(\exp(z_1)+, \cdots, \exp(z_1)))\\
    &\leq   \log(N \exp(\max\{z_1,\cdots, z_N\}))\\
    &= \max\{z_1,\cdots, z_N\}+\log(N)     \text{.}   
    \end{aligned}
\end{equation}
\end{proof}
Therefore, we can utilize the maximum function to approximate the log-sum-exp function. The approximation error is small when the exponential of the largest element is much larger than the exponential of other elements. It is the case for docking scores since they are negative in most cases. Then \eqref{eq:ds1} can be approximated by $-S(x_{ij},y_i)-\max_{y_k} (-S(x_{ij},y_k))$. Analogous to existing docking scores, we define the initial Delta Score as
\begin{align}\label{eq:ds2}
    S(x_{ij},y_i)-\min_{y_k} S(x_{ij},y_k) ,
\end{align}
which has the same unit as the utilized docking score $S(x,y)$ and is also the smaller the better. 
Notice that to calculate \eqref{eq:ds2} for all $i=1\cdots n$, $j=1,\cdots m_i$, we need to perform docking for molecules with all possible pockets which is of quadratic complexity $O(n^2)$. To address this challenge, we provide a computationally efficient estimation of the initial Delta Score 
\begin{align}\label{eq:ds3}
   S(x_{ij},y_i) - S(x_{ij},y_k) , \text{where $k$ is randomly sampled in $\{1, 2, \ldots, n\}$ with $k \neq i$}.
\end{align}
Finally, we arrive at the definition of the Delta Score for a given pocket: For any pocket $y_i$, we sample  $k \in \{1, 2, \ldots, n\}$ with $k \neq i$: 
\begin{equation}
\label{formula:delta score}
\begin{split}
\textbf{Delta Score}(y_i) \triangleq \frac{1}{m_i} \sum_{j=1}^{m_i} (S(x_{ij},y_i) - S(x_{ij},y_k)).
\end{split}
\end{equation}

 In fact, \eqref{eq:ds3} serves as an unbiased estimation of the quantity $S(x_{ij},y_i) - \frac{1}{n-1}\sum_{k\neq i} S(x_{ij},y_k)\leq S(x_{ij},y_i)-\min_{y_k} S(x_{ij},y_k)$, consistently yielding lower values than \eqref{eq:ds2}.
Both formulations can be intuitively explained by emphasizing the binding affinity of molecules with their specific target protein pockets, as opposed to multiple similar pockets without selectivity. A smaller value of \eqref{eq:ds2} indicates the molecule $x_{ij}$, generated by the model specifically for target $y_i$, has a high affinity towards target $y_i$itself, relative to any alternative target $y_k$. Whereas \eqref{eq:ds3} relaxes this constraint, allowing for comparisons of affinity relative to the average affinity across all other targets $y_k$.

\section{SE(3)-Equivariant GNN}
In Section 4, we denote molecules and pockets as $X = [c^X, v^X]$ and $Y = [c^Y, v^Y]$ respectively. 
\subsection{Parameterization of diffuison process}
Denote the GNN network as $\phi_\theta$. We model the whole complex as a graph $[X, Y]$ and let the neural network predict $[x_0, v_0]$ at each time step:
\begin{equation}
    [\hat{x_0},\hat{v_0}] = \phi_\theta(X_t, t, Y) = \phi_\theta([c^X_t, v^X_t], t, Y).
\end{equation}
At the l-th layer, the atom hidden embedding $\textbf{h}$ and coordinates $c$ are updated alternately as follows:

\begin{align}
&\textbf{h}i^{l+1} = \textbf{h}i^l + \sum_{j\in\mathcal{V},j\neq i} f_{\textbf{h}}(d^l_{ij},\textbf{h}i^l,\textbf{h}j^l,e{ij};\theta_{\textbf{h}})  \text{,}\notag \\
&c_i^{l+1} = c_i^l + \sum_{j\in\mathcal{V},j\neq i}(c_i^l-c_j^l)f_c(d^l_{ij},\textbf{h}i^{l+1},\textbf{h}j^{l+1},e{ij};\theta_x) \cdot \mathbb{1}_{mol}\text{.}
\end{align}
, where $d_{ij}$ is the euclidean distance between atoms $i$ and $j$ along with an additional feature $e_{ij}$ indicating the connection type (protein-protein, ligand-ligand, or protein-ligand). The mask $\mathbb{1}_{mol}$ is applied to prevent updates on protein atom coordinates. The initial atom hidden embedding $\textbf{h}_0$ is obtained through a linear embedding layer that encodes atom information. The final atom hidden embedding $\textbf{h}_L$ is then passed through a MLP and a softmax function to obtain the predicted value $\hat{v_0}$.

\subsection{Parameterization of SBE-Encoder}
We also adopt 3D-Equivariant graph neural network to parameterize $\text{SBE-Enc}_X$ and $\text{SBE-Enc}_Y$. Different from B.1, molecules and pockets are encoded separately instead of forming a single full graph. Taking  $\text{SBE-Enc}_X$ as an example, $\text{SBE-Enc}_Y$ follows the same principle.

At the l-th layer, we do not update the atom coordinates but only the hidden embedding $\textbf{h}$:
\begin{align}
\textbf{h}i^{l+1} = \textbf{h}i^l + \sum_{j\in\mathcal{V},j\neq i} f_{\textbf{h}}(d^l_{ij},\textbf{h}i^l,\textbf{h}j^l,e{ij};\theta_{\textbf{h}})\text{.}
\end{align}
In the final layer, we obtain the encoding of the entire molecule by performing average pooling over all the atoms:
\begin{equation}
f(x) = \text{SBE-Enc}_X(X) = \frac{1}{N_X}\sum_{i=1}^{N_X}\textbf{h}_i^L\text{.}
\end{equation}



\section{Visualizations}
\label{appendix: visualization}

\begin{figure}[ht]
\centering
\begin{subfigure}{.5\textwidth}
  \centering
  \includegraphics[width=.95\linewidth]{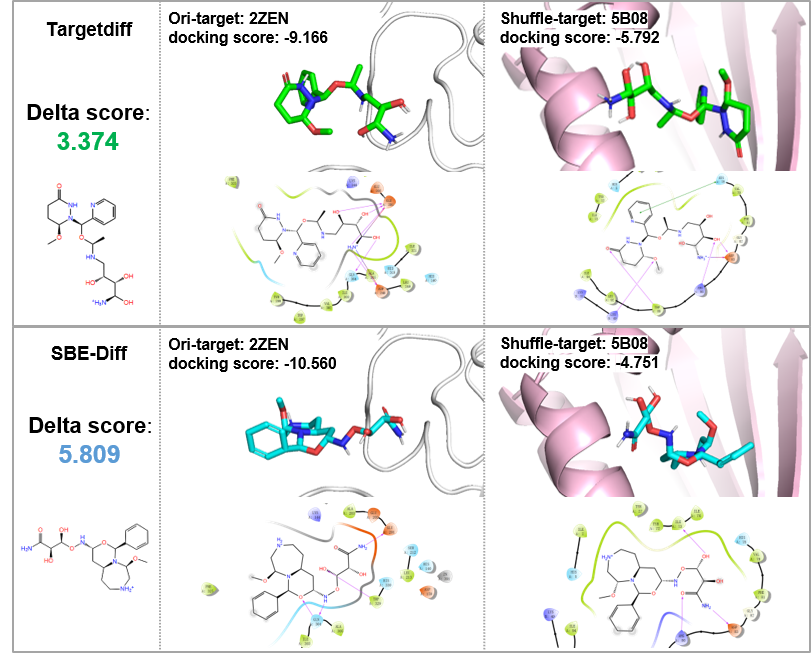}
  \caption{}
  \label{fig:sub1}
\end{subfigure}%
\begin{subfigure}{.5\textwidth}
  \centering
  \includegraphics[width=.95\linewidth]{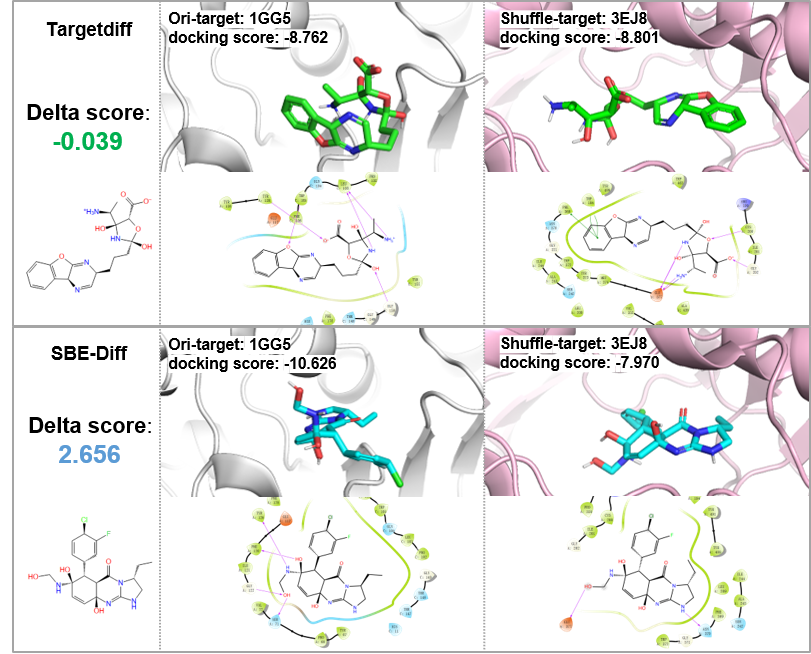}
  \caption{}
  \label{fig:sub2}
\end{subfigure}
\newline
\begin{subfigure}{.5\textwidth}
  \centering
  \includegraphics[width=.95\linewidth]{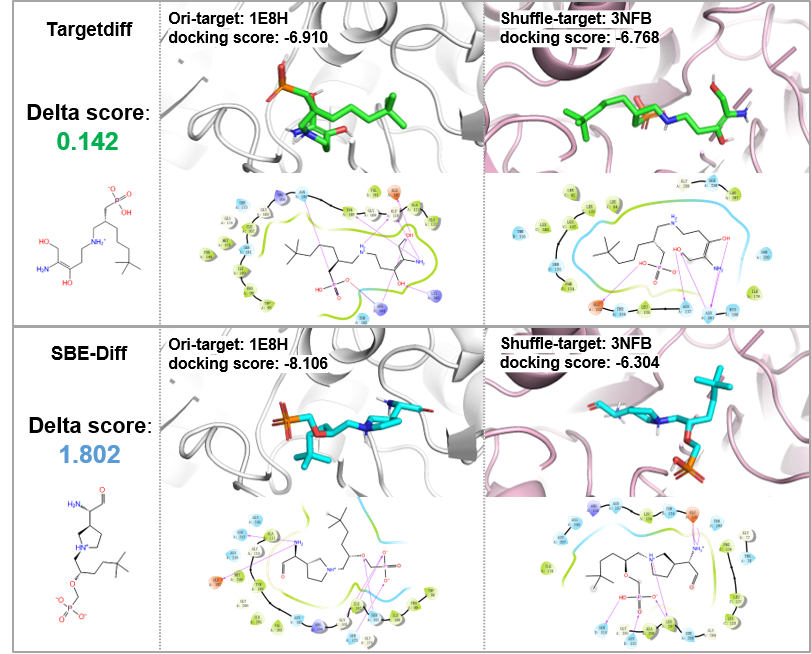}
  \caption{}
  \label{fig:sub3}
\end{subfigure}%
\begin{subfigure}{.5\textwidth}
  \centering
  \includegraphics[width=0.95\linewidth]{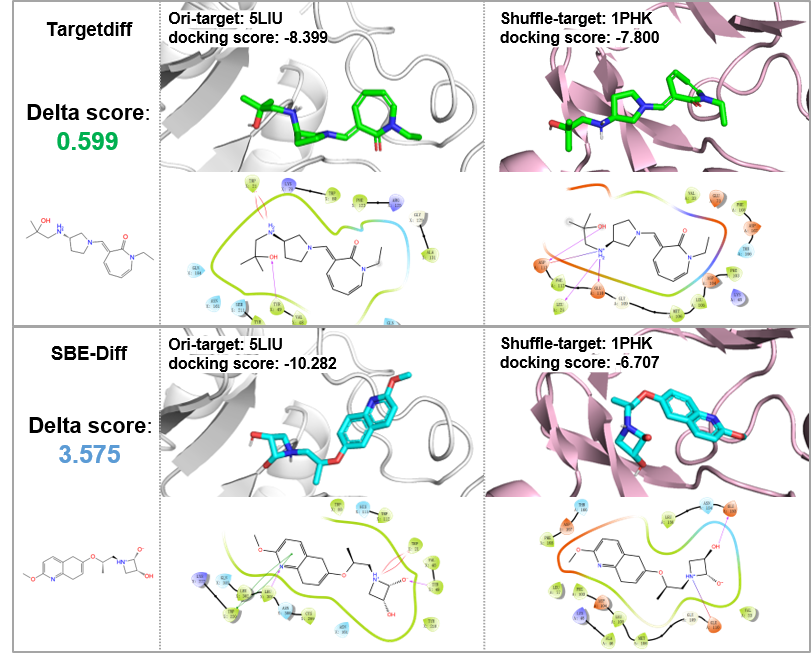}
  \caption{}
  \label{fig:sub4}
\end{subfigure}
\caption{Visualizations of molecules generated by SBE-Diff versus molecules generated by TargetDiff. Molecules are docked with their original targets and shuffled targets.}
\label{fig:test}
\end{figure}

\end{document}